\documentclass[12pt]{amsart}

\usepackage{amsmath,amsthm,amssymb,epsfig}
\usepackage{hyperref}
\usepackage{verbatim}
\usepackage{tikz}
\usepackage{float}
\usepackage{mathtools}
\usepackage{subcaption}
\usepackage{algorithm}
\usepackage{algorithmic}
\usepackage{MnSymbol}
\usepackage{cleveref}
\usepackage{graphicx}
\usepackage[left=2.3cm,top=2.3cm,right=2.3cm,bottom=2.3cm]{geometry}
\usepackage[normalem]{ulem}

\newtheorem{theorem}{Theorem}
\newtheorem*{theorem*}{Theorem}
\newtheorem{lemma}[theorem]{Lemma}

\theoremstyle{definition}

\newcommand{\qedclaim}{\hfill $\diamond$ \medskip}

\newcommand{\HC}{{\rm HC}}
\newcommand{\clone}[1]{{#1}'}

\newcommand{\walk}{walk }

\newcommand{\hyperpath}{path }
\newcommand{\hyperpaths}{paths }

\theoremstyle{plain}

\begin{document}

\title{The iterated local transitivity model for hypergraphs}\thanks{The authors are funded by NSERC. The third author was funded by an NSERC Postdoctoral Fellowship.}

\author[N.C.\ Behague]{Natalie C.\ Behague}
\author[A.\ Bonato]{Anthony Bonato}
\author[M.A.\ Huggan]{Melissa A.\ Huggan}
\author[R. Malik]{Rehan Malik}
\author[T.G.\ Marbach]{Trent G.\ Marbach}
\address[A1, A2, A3, A4, A5]{Ryerson University, Toronto, Canada}
\email[A1]{(A1) nbehague@ryerson.ca}
\email[A2]{(A2) abonato@ryerson.ca}
\email[A3]{(A3) melissa.huggan@ryerson.ca}
\email[A4]{(A4) rtmalik26@ryerson.ca}
\email[A5]{(A5) trent.marbach@ryerson.ca}

\begin{abstract}
Complex networks are pervasive in the real world, capturing dyadic interactions between pairs of vertices, and a large corpus has emerged on their mining and modeling. However, many phenomena are comprised of polyadic interactions between more than two vertices. Such complex hypergraphs range from emails among groups of individuals, scholarly collaboration, or joint interactions of proteins in living cells. Complex hypergraphs and their models form an emergent topic, requiring new models and techniques.

A key generative principle within social and other complex networks is transitivity, where friends of friends are more likely friends. The previously proposed Iterated Local Transitivity (ILT) model incorporated transitivity as an evolutionary mechanism. The ILT model provably satisfies many observed properties of social networks, such as densification, low average distances, and high clustering coefficients.

We propose a new, generative model for complex hypergraphs based on transitivity, called the Iterated Local Transitivity Hypergraph (or ILTH) model.
In ILTH, we iteratively apply the principle of transitivity to form new hypergraphs. The resulting model generates hypergraphs simulating properties observed in real-world complex hypergraphs, such as densification and low average distances. We consider properties unique to hypergraphs not captured by their 2-section. We show that certain motifs, which are specified
subhypergraphs of small order, have faster growth rates in ILTH hypergraphs than in random hypergraphs with the same order and expected average degree.
We show that the graphs admitting a homomorphism into the 2-section of the initial hypergraph appear as induced subgraphs in the 2-section of ILTH hypergraphs.
We consider new and existing hypergraph clustering coefficients, and show that these coefficients have larger values in ILTH hypergraphs than in comparable random hypergraphs.
\end{abstract}


\maketitle

\section{Introduction}

Complex networks are an effective paradigm for pairwise interactions between objects in real-world systems. Such networks capture dyadic interactions in many phenomena, ranging from friendship ties in Facebook, to Bitcoin transactions, to interactions between proteins in living cells. Complex networks evolve via a number of mechanisms such as preferential attachment or copying that predict how links between vertices are formed over time. \emph{Structural balance theory} cites mechanisms to complete triads (that is, subgraphs consisting of three vertices) in social and other complex networks \cite{ek,he}. A central mechanism in balance theory is \emph{transitivity}: if $x$ is a friend of $y,$ and $y$ is a friend of $z,$ then $x$ is a friend of $z$; see, for example, \cite{scott}.

The \emph{Iterated Local Transitivity} (\emph{ILT}) model introduced in \cite{ilt,ilt1} and further studied in \cite{ilm,ildt,mason}, simulates structural properties in complex networks emerging from transitivity. Transitivity gives rise to the notion of \emph{cloning}, where an introduced vertex $x$ is adjacent to all of the neighbors of some pre-existing vertex $y$. Note that in the ILT model, the vertices have local influence within their neighbor sets. Although graphs generated by the model evolve over time, there is a memory of the initial graph hidden in the structure. The ILT model simulates many properties of social networks. For example, as shown in \cite{ilt}, graphs generated by the model densify over time and exhibit bad spectral expansion. In addition, the ILT model generates graphs with the small-world property, which requires graphs to have low diameter and high clustering coefficient compared to random graphs with the same number of vertices and expected average degree.

Dyadic relationships do not always fully capture the dynamics of interactions between larger groups of vertices. For example, interactions among groups of vertices occur in scholarly collaborations, tags attached to the same web post, or metabolic interactions between more than two reactants. In these examples, a polyadic view of interactions is more accurate, giving rise to hypergraphs. A \emph{hypergraph} is a discrete structure with vertices and \emph{hyperedges}, which consists of sets of vertices. Graphs are special cases of hypergraphs, where each hyperedge has cardinality two. While hypergraph theory is less developed than graph theory, it is an emerging topic in the study of complex, real-world systems; see, for example, \cite{benson,benson1,do,estrada,GallagherGoldberg2013,lee,zn}. For a recent article discussing the important role of hypergraphs and other higher-order methods for studying complex networks, see \cite{news}.

In the present paper, we consider a deterministic model for complex hypergraph networks based on transitivity. The model is analogous to the ILT model, although it has its own unusual features. While every hypergraph can be reduced to its 2-section graph, replacing each hyperedge by a clique, not all hypergraph properties are captured by the 2-section. As we demonstrate, the ILT hypergraph model we introduce has properties not evident in its 2-section. Further, the model simulates several properties, such as clustering and motif evolution, more robustly when compared to random hypergraphs with analogous characteristics. For simplicity, we consider throughout $k$-\emph{uniform} hypergraphs, where each hyperedge has cardinality $k$ for a fixed positive integer $k\ge 2.$

The \emph{Iterated Local Transitivity Hypergraph} (\emph{ILTH}) model is defined formally as follows. The model is deterministic and generates $k$-uniform hypergraphs over discrete time-steps. The sole parameter of the model is the initial $k$-uniform hypergraph $H=H_{0}.$ For a nonnegative integer $t$, the hypergraph $H_{t}$ represents the hypergraph at time-step $t$. To form $H_{t+1}$, for each $x \in V(H_{t}$), add a new vertex $x'$ called the \emph{clone} of $x$. We refer to $x$ as the \emph{parent} of $x',$ and $x'$ as the \emph{child} of $x.$ For every hyperedge $e$ of $H_t$ containing $x$, we  add the hyperedge $e'$ to $H_{t+1}$ formed by replacing $x$ with $x'$.
Observe that $e' = (e \setminus \{x \}) \cup \{ x' \}$; we simply write $e'=e-x+x'.$ Note that all existing hyperedges in $H_t$ are also included in $H_{t+1}$. See Figure~\ref{fig1}. We refer to $H_t$ as an \emph{ILTH hypergraph}, and we sometimes write $H_t=$ILTH$_t(H)$ to emphasize the initial hypergraph $H.$ Note that ILTH$_t(H)$ is $k$-uniform for all $t\ge 0.$ We sometimes refer to the formation of the hypergraphs $H_t$
as the \emph{ILTH process}.

The clones form an independent set in $H_{t+1}$, resulting in a doubling of the order of $H_t.$  Unlike in the ILT model, a clone and its parent are not in a hyperedge. For a vertex $x$ in $H_t$, we will sometimes use the notation $x^*$ to mean any \emph{descendant} of $x$; that is, $x^*$ is either $x$ or $x'$ in $H_{t+1}$. Similarly, if $e$ is a hyperedge in $H_t$, then $e^*$ represents one of the descendant hyperedges $e$ or $e - x + x'$ in $H_{t+1}$.
\begin{figure}[h]
\begin{center}
\epsfig{figure=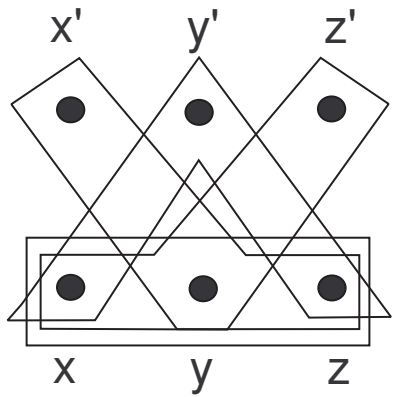,scale=1.5}
\caption{The ILTH model with $H_0$ a hyperedge with $k=3.$}\label{fig1}
\end{center}
\end{figure}

As we will demonstrate, the ILTH model simulates many properties observed in complex hypergraphs, including the small-world property and motif counts. In Section~\ref{2sec}, we derive a densification power law for ILTH hypergraphs, and show that distance and spectral properties follow by properties of the 2-section. We then consider subhypergraphs and motifs in Section~\ref{secmotif}. Motifs are certain hypergraphs with a small number of vertices and hyperedges. In \cite{lee}, it was shown that several real-world, complex hypergraphs have motif counts dramatically higher than comparable random hypergraphs. We show that for certain motifs arising in $k$-uniform hypergraphs from the list in \cite{lee} of 26 motifs formed from three hyperedges, ILTH has a provably higher count than in a random hypergraph with the same average degree. We prove that the 2-section contains isomorphic copies of all graphs admitting a homomorphism to the 2-section of $H_0$ in Theorem~\ref{mainin} and contains only such graphs; as a consequence, certain motifs will be excluded in the ILTH process unless they appear in $H_0$.

In Section~\ref{secclus}, we provide a rigorous analysis of various clustering coefficients for ILTH hypergraphs. Our study of clustering coefficients further validates the small-world property of ILTH hypergraphs, and leads to interesting combinatorial analysis. We consider two clustering coefficients $\HC_1$ and $\HC_2$ and their asymptotic order in ILTH. The clustering coefficient $\HC_1$ was first studied in \cite{estrada}. We introduce the new parameter $\HC_2$ that is a variant of one that first appeared in \cite{zn}, although we argue it is more natural and amenable to analysis. In the case of $\HC_1,$ we show that these clustering coefficients provide higher clustering than is expected in random hypergraphs with the same average degrees. We show an analogous result for $\HC_2$ in a variation of the ILTH model, where clones and parents are adjacent. We finish with a summary of our results along with open problems on the ILTH model.

Throughout the paper, we consider finite, simple, undirected graphs and hypergraphs. For a general reference on graph theory, see \cite{west}. For a reference on hypergraphs, see \cite{berge,berge1}. For background on social and complex networks, see \cite{bbook,bt,chung1}. We define terms and notation for hypergraphs when they first appear throughout the article.

\section{Densification, eigenvalues, and distances}\label{2sec}

Many examples of complex networks \emph{densify} in the sense that the ratio of their number of edges to vertices tends to infinity over time; see \cite{les1}. In this section, we show that the ILTH model always generates hypergraphs that densify, and we give a precise statement below of its densification power law.

Let $n(t)$ be the number of vertices in $H_t$ and let $e(t)$ be the number of hyperedges in $H_t$, respectively. We establish elementary though important recursive formulas for these parameters.
\begin{theorem}\label{ne}
For a nonnegative integer $t,$ we have the following.
\begin{enumerate}
\item $n(t) = 2^tn(0).$
\item  $e(t) = (k+1)^t e(0).$
\end{enumerate}
In particular, we have that $e(t) = \Theta\left(n(t)^{\log_2{(k+1)}}\right)$.
\end{theorem}

\begin{proof}
For item (1), for each vertex $v$ in  $H_{t},$ there are two vertices $v$ and $v'$ in $H_{t+1}$. Hence, $n(t+1) = 2n(t)$.

For item (2), notice that for each hyperedge $e$ in $H_{t},$ we add to $H_{t+1}$ the hyperedge $e$ and each of  the $k$ hyperedges $e - x + x'$ where $x$ is a vertex in $e$. We then have that $e(t+1) = (k+1)e(t)$ for all $t$. The result follows.
\end{proof}

As a consequence, the average vertex degree of ILTH$_t(H)$ is given by $$\frac{ke(t)}{n(t)} =\left(\frac{k+1}{2}\right)^t \frac{ke(0)}{n(0)},$$ which increases exponentially with $t$. Hence, we have a densification power law for ILTH hypergraphs.

We next turn to the 2-section of ILTH hypergraphs. For this, we consider a variant on the ILT model for graphs, which we call ${\rm ILT}'$.
Given a graph $G=G_0$, iteratively construct ${\rm ILT}'_t(G),$ where $t \ge 1$ as follows. Suppose that we have ${\rm ILT}'_t(G)$. For each $v \in V({\rm ILT}'_t(G))$, the vertices $v$ and $v'$ are included in ${\rm ILT}'_{t+1}(G)$. For each $uv \in E({\rm ILT}'_t(G))$, the edges $uv$, $uv'$ and $u'v$ are included in ${\rm ILT}'_{t+1}(G)$.
We have the following lemma, whose proof is immediate.
\begin{lemma}\label{lem:2-section}
For a nonnegative integer $t,$ we have that ${\rm ILT}'_t(G)$ is the 2-section of ${\rm ILTH}_t(H)$.
\end{lemma}

We use the notation $n_t$ and $e_t$ for the order and size of ${\rm ILT}'_t(G)$. Observe that $n_t = 2^tn_0$ and $e_t=3^te_0$ edges. An implication of Lemma~\ref{lem:2-section} is that any hypergraph property that depends solely on the 2-section behaves the same way for the hypergraph model ${\rm ILTH}$ as it does for the graph model ILT$'$. Such properties are not truly exploiting the hypergraph structures evident in ${\rm ILTH}$. We briefly discuss some of these properties, including the adjacency matrix, the diameter, and the average distance.

The adjacency matrix $A(H)$ for a hypergraph $H$ has rows and columns indexed by the vertices of $H$ and entry $1$ if $u \ne v$ and there is some hyperedge of $H$ containing both $u$ and $v$, and $0$ otherwise.  It is evident that this is the same as the adjacency matrix of the 2-section of $H$. In particular, to analyse the adjacency matrix of ${\rm ILTH}_t(H)$ we need only consider the adjacency matrix of ${\rm ILT}'_t(G)$, where $G$ is the 2-section of $H$.

If ${\rm ILT}'_t(G)$ has $n \times n$ adjacency matrix $A$, then ${\rm ILT}'_{t+1}(G)$ has $2n \times 2n$ adjacency matrix
$$
\left(\begin{matrix}
A & A \\
A & \textbf{0}
\end{matrix} \right),
$$
where $\textbf{0}$ is the $n \times n$ all-zeros matrix. It is straightforward to verify that if $A$ has eigenvalue $\rho$ with associated eigenvector  $\mathbf{v},$ then $
\left(\begin{smallmatrix}
A & A \\
A & \textbf{0}
\end{smallmatrix} \right)$  has eigenvalues $\frac{1 \pm \sqrt{5}}{2}\rho$ with associated eigenvectors
$
\left(\begin{smallmatrix}
\frac{1 \pm \sqrt{5}}{2} \mathbf{v} \\
\mathbf{v}
\end{smallmatrix} \right)
$.
In particular, given the eigenvalues for the graph $G$, one can calculate the eigenvalues for ${\rm ILT}'_t(G)$.

We next consider distance in ILTH hypergraphs. A \emph{\walk}of length $k$ connecting two vertices $u$ and $v$ in a hypergraph is a sequence of hyperedges $e_1,e_2,\dots, e_k$ such that $u \in e_1$, $v \in e_k$ and $e_i \cap e_{i+1} \ne \emptyset$, for all $1 \le i < k$. We say that the \emph{distance} between two vertices $u,v,$ written $d(u,v),$ is the minimum length of a \walk connecting $u$ and $v$. This is the same as the distance between two vertices $u$ and $v$ in the 2-section of the hypergraph. In particular, to analyze distances within  ${\rm ILTH}_t(H)$ we could only consider distances in ${\rm ILT}'_t(G)$, where $G$ is the 2-section of $H$, but it is equally convenient to analyse ILTH directly.

Consider vertices $u,v$ in $H_t$ with $u \ne v$. Let $d = d(u,v)$ and let  $e_1,e_2,\dots, e_d$ be a minimum length \walk connecting them. We then have that in $H_{t+1}$,
\begin{enumerate}
\item $d(u,v) = d$, using the \walk  $e_1,e_2,\dots, e_d$;
\item $d(u,v') = d$, using the \walk  $e_1,e_2,\dots, e_d-v+v'$;
\item $d(u',v) = d$, using the \walk  $e_1-u+u',e_2,\dots, e_d$;
\item $d(u',v')=d$ if $d \ge 2$, using the \walk  $e_1-u+u',e_2,\dots, e_d-v+v'$; and
\item $d(u',v')=2$ if $d = 1$, using the \walk  $e_1-u+u' , e_1-v+v'$, so long as $k \ge 3$.
\end{enumerate}
Note that in the case of the final item, there is no \walk of length one as there is no hyperedge containing two clones. In the other cases, there can be no {walks} 
of length less than $d$ else the predecessors of these edges would form a \walk from $u$ to $v$ in $H_t$ of length less than $d$.

The \emph{diameter} of a hypergraph is the maximum distance between any pair of vertices. We find immediately that the diameter of $H_{t+1}$  is the maximum of 2 and the diameter of $H_t$, and, iterating this, is the maximum of 2 and the diameter of $H_0$. In either case, the diameter is a constant, independent of $t$.

To end this section, we determine the average distance between any pair of vertices in $H_t$.
Let $W(t)$ be the sum of the distances in $H_t$ or \emph{Wiener index}, written
 $$W(t) = \sum_{u,v \in V( H_t)} d(u,v).$$
Assuming that $H_0$ has no isolated vertices and so $H_t$ has no isolated vertices for all $t\ge1$, by our calculations pertaining to distances above, we obtain that:
\begin{align*}
W(t+1)&=  \sum_{u,v \in V( H_{t+1})} d(u,v) \\
&= \sum_{u \ne v \in V(H_t )}  d(u,v) + d(u',v) + d(u,v') + d(u',v') + \sum_{u \in V( H_t)} d(u,u') + d(u',u) \\
&= 4 \left( \sum_{u,v \in V( H_t) }  d(u,v) \right) + |\{ u \ne v \in V(H_t): d(u,v) = 1 \}| + 4n(t) \\
 &= 4 W(t) + 2e(t) + 4n(t).
\end{align*}

Solving this recurrence gives that
\begin{align*}
W(t) &= 4^t(W(0) + 2e(0) + 2n(0)) - 2e(t)  - 2n(t) \\
&= 4^t(W(0) + 2e(0) + 2n(0)) - 2 \cdot 3^te(0) - 2^{t+1}n(0).
\end{align*}
Thus, the average distance is given by
$$\frac{2W(t)}{n(t)(n(t)-1))} = \frac{ 4^t(2W(0) + 4e(0) + 4n(0)) - 4\cdot 3^te(0) - 2^{t+1}2n(0)}{4^tn(0)^2 - 2^tn(0)},$$
which tends to $\frac{2W(0) + 4e(0) + 4n_0}{n(0)^2}$ as $t$ tends to infinity. We therefore have that ILTH hypergraphs exhibit a constant average distance, as is found in many real-world hypergraphs; see \cite{do}.

\section{Subhypergraphs and motifs}\label{secmotif}

We next consider subhypergraphs of the ILTH model, and our first approach is to consider the induced subgraphs of the 2-section. In Theorem~\ref{mainin}, it is shown that a graph appears in the 2-section of an ILTH hypergraph exactly when it admits a homomorphism to the 2-section of $H_0.$ The theorem guarantees the absence of many kinds of induced subhypergraphs; for example, no hypergraph clique appears in an ILTH hypergraph with larger order than $H_0.$ We then turn to counting certain small order subhypergraphs, or motifs. Motifs are important in complex networks, as they are one measure of similarity for graphs. For example, the counts of $3-$ and $4-$vertex subgraphs gives a similarity measure for distinct graphs; see \cite{mgeop,pr} for implementations of this approach using machine learning. Hypergraph motifs were studied by several authors; see for example, \cite{aksoy,benson1,lee}. In \cite{lee}, motif counts were analyzed across various real-world complex hypergraphs and compared to random hypergraphs.    We show in this section that in ILTH hypergraphs, the growth rate for certain motifs is higher than in comparable random hypergraphs.

\subsection{Induced subgraphs of the 2-section}

For all $t\ge 0,$ $H_t$ is an induced subhypergraph of $H_{t+1}.$ There exists a homomorphism $f_t$ from $H_{t+1}$ to $H_{t}$ by mapping each clone to its parent, and fixing all other vertices. Note that $F_t = f_1 \circ f_2 \circ \dots \circ f_t$ is a homomorphism from $H_t$ to $H_0.$ As a result, the clique and chromatic numbers of $H_t$ are bounded above by those of $H_0.$
This observation puts limitations on the kinds of subgraphs that $H_t$ contains. For additional background on graph homomorphisms, the reader is directed to \cite{hn}.

The \emph{age} of a hypergraph is its set of isomorphism types of induced subhypergraphs. As each $F_t$ is a homomorphism, we have that no $H_t$ contains $k$-uniform cliques larger than those in $H_0$. In particular, the set of ages of an {\rm ILTH} hypergraph does not contain all hypergraphs. This contrasts with the ILT model, where all graphs occur in the set of ages of ILT-graphs; see \cite{ilm}.

Characterizing the ages of {\rm ILTH} hypergraphs remains an open problem. The next result solves the analogous problem for the ages of $2$-sections of {\rm ILTH} hypergraphs. For a fixed graph $G$ and family of graphs $\mathcal{G}$, we say that $\mathcal{G}$ is $G$-\emph{hom-universal} if the set of ages of $\mathcal{G}$ consists of all finite graphs admitting a homomorphism to $G.$

\begin{theorem}\label{mainin}
A graph $G$ admits a homomorphism to $G(H_0)$ if and only if $G$ is an induced subgraph of $G(H_t),$ for some integer $t \ge 0$ and where $G(H_0)$ is the 2-section of $H_0$. In particular, the set of ages of 2-sections of hypergraphs in $\mathrm{{\rm ILTH}}(H_0)$ is $G(H_0)$-hom-universal.
\end{theorem}

\begin{proof} The reverse direction follows since for an induced subgraph $G$ of $G(H_t),$ the inclusion map is a homomorphism from $G$ to $G(H_t).$ Composing with $F_t$ gives a homomorphism from $G$ to $G(H_0).$

For the forward direction, suppose that $G$ admits a homomorphism $f$ to $G(H_0)$. Let $u,v\in V(G)$ be two vertices such that $f(u)=f(v)$.
Define the homomorphism $f'$ to $G(H_1)$ as $f'(x) = f(x)$ if $x \neq v,$ and $f(v)$ is the clone of the vertex $f(u)$.
We then note that the number of vertices in the codomain of $f'$ is one larger than the number of vertices in the codomain of $f$.
We may repeat this procedure until we find an injective homomorphism $f_i$ from $G$ to $G(H_i)$, for some $i \ge 0$.

Suppose that there are two vertices $u,v$ in $G$ which are not neighbors but such that $f_i(u) f_i(v)$ is an edge in $G(H_i)$. We can define a new injective homomorphism $f_{i+1}$ to $G(H_{i+1})$ by $f_{i+1}(x) = f_{i}(x)$ if $x \notin \{u,v\}$, $f_{i+1}(u)$ is the clone of $f_i(u)$, and $f_{i+1}(v)$ is the clone of $f_i(v)$.
We then have that the induced subgraph of $f_i(G)$ and $f_{i+1}(G)$ differ only in the edge $f_i(u) f_i(v)$, as this edge does not exist in $f_{i+1}(G)$.
We can repeat this procedure to construct an injective homomorphism $f_j$ from $G$ to $G(H_j)$ for some $j,$ with the property that for all $u,v \in V(G)$ if $f_j(u) f_j(v)$ is an edge in $G(H_j)$, then $uv$ is an edge in $G$. Hence, the subgraph induced by the vertices in $f_j(G)$ in $G(H_j)$ is isomorphic to $G$.
\end{proof}

\subsection{Motifs}
We now turn to counting motifs, which are certain types of subhypergraphs. In \cite{lee}, 26 distinct motifs were studied for three interacting hyperedges $e_1,$ $e_2,$ and $e_3.$ Motif counts may be viewed as a similarity measure for hypergraphs, such as when we are comparing real-world hypergraphs and synthetic ones derived from models.

The different types of motifs emerge by considering which of the following seven regions are nonempty:  $$e_1 \setminus (e_2 \cup e_3),e_2 \setminus (e_1 \cup e_3),e_3 \setminus (e_1 \cup e_2),e_1\cap e_2 \setminus e_3,e_2\cap e_3 \setminus e_1,e_1\cap e_3 \setminus e_2, e_1 \cap e_2 \cap e_3.$$ We may compactly reference motifs by a binary sequence $$i_1i_2i_3i_4i_5i_6i_7,$$ so that for all $j$, $i_j =1$ exactly when there is at least one element in the corresponding region. We refer to the different motifs as \emph{motif types}. See Figure~\ref{fig:m11} for an example.
\begin{figure}[h]
\begin{center}
\epsfig{figure=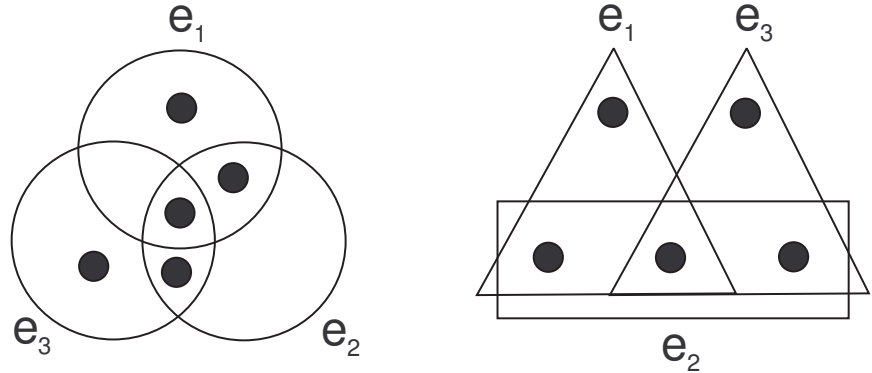,scale=1.5}
\caption{The motif type 11 or 1011101. On the left, we represent this motif via a Venn diagram, where the vertex in a region implies it is nonempty. On the right, we have an example of a $3$-uniform hypergraph realizing this motif.}\label{fig:m11}
\end{center}
\end{figure}
We may generalize this notation to a tuple of nonnegative integers, quantifying the number of elements in each region. The \emph{cardinality vector} of a motif composed of the three hyperedges $e_1, e_2, e_3$ is defined as the $7$-tuple:
\begingroup\makeatletter\def\f@size{10} {$$\check@mathfonts (a,b,c,d,e,f,g)= (|e_1 \setminus (e_2 \cup e_3)|,|e_2 \setminus (e_1 \cup e_3)|, |e_3 \setminus (e_1 \cup e_2)|, |e_1\cap e_2 \setminus e_3|, |e_2\cap e_3 \setminus e_1|, |e_1\cap e_3 \setminus e_2|, |e_1 \cap e_2 \cap e_3|).$$} \endgroup Note that a motif contains $a+b+c+d+e+f+g$ vertices. Further, we have that $|e_1| = a+d+f+g = k$, $|e_2| = b+d+e+g = k$, and $|e_3| = c+e+f+g = k$.

In general hypergraphs, there are 26 non-isomorphic motif types; however, we note that only 11 motif types occur in $k$-regular hypergraphs. With numbering taken from \cite{lee}, these motif types are:
\begin{enumerate}
\item Motif type 2: 1110001,
\item Motif type 6: 1110101,
\item Motif type 11: 1011101,
\item Motif type 12: 1111101,
\item Motif type 13: 0001111,
\item Motif type 14: 1001111,
\item Motif type 15: 1011111,
\item Motif type 16: 1111111,
\item Motif type 24: 1001110,
\item Motif type 25: 1011110,
\item Motif type 26: 1111110.
\end{enumerate}
We keep the numbering from \cite{lee} for brevity; for example, we refer to motif 11 rather than 1011101. We focus on these motif types since they always occur in the ILTH model and have higher counts when compared to random hypergraphs, as we describe below. Interestingly, motifs 11 and 12 are more prevalent in the co-authorship hypergraphs compared to random hypergraphs, as shown in \cite{lee}. The same conclusion holds for motif 16 for tag hypergraphs. These observations lend credence to the view that ILTH hypergraphs simulate properties observed in real-world, complex hypergraphs.

Let $\alpha_i$ be the maximum number of vertices that can occur in a motif of type $i$ in a $k$-uniform hypergraph. Each value of $\alpha_i$ can be calculated explicitly, and each calculation is straightforward. For example, we may calculate $\alpha_{14}$ as follows. Suppose that the motif in question has cardinality vector $(a,0,0,d,e,f,g)$. Without loss of generality we have that $$k = a+d+f+g = d+e+g = e+f+g,$$ as each hyperedge contains $k$ vertices. It therefore immediately follows that $d=f$. The total number of vertices is $$a+d+e+f+g = k+(k-d-g),$$ which is maximized when $d=f=g=1$ (as $d,g,f>0$ in a motif of type $14$), yielding $\alpha_{14} = 2k-2$. As the remaining calculations are similar to the above, we omit them, and present the results in Table~\ref{tab:alphas}.
\begin{center}
\begin{table}[h]
\begin{tabular}{|l|l|l|l|l|l|l|l|l|l|l|l|l|l|l|l|l|l|l|l|l|l|l|l|l|l|}
\hline
$i$&2&6&11&12&13&14&15&16&24&25&26\\ \hline
$\alpha_i$&$3k-2$&$3k-3$&$2k-1$&$3k-4$&$\lfloor \frac{3}{2}k-1\rfloor$&$2k-2$&$2k-2$&$3k-5$&$2k-1$&$2k-1$&$3k-3$\\ \hline
\end{tabular}
\caption{The maximum number of vertices $\alpha_i$ in a motif of type $i$ possible in a $k$-uniform hypergraph.}
\label{tab:alphas}
\end{table}
\end{center}

\begin{lemma} \label{lm:fromCardVect}
If $H_t$ contains $x$ motifs of type $i$ with cardinality vector  $(a,b,c,d,e,f,g)$, then $H_{t+1}$ contains at least $x ( g + (c+1)d + (b+1)f + (a+1)e + (a+1)(b+1)(c+1) )$ motifs of type $i$ with cardinality vector  $(a,b,c,d,e,f,g)$.
\end{lemma}
\begin{proof}
For a motif in $H_t$ of type $i$ with cardinality vector  $(a,b,c,d,e,f,g)$ formed by the hyperedges $e_1, e_2, e_3$, we choose a set $S$ of up to three vertices contained in the motif to clone such that each hyperedge of the motif contains at most one cloned vertex.
Consider the motif in $H_{t+1}$ formed by the hyperedges $e'_1, e'_2, e'_3$, where $e'_i$ is the hyperedge obtained from $e_i$ by replacing each vertex that is also in $S$ with its clone and leaving other vertices unchanged.
This motif is of type $i$ and has cardinality vector  $(a,b,c,d,e,f,g)$. Each motif developed in this way is unique.
We must therefore find how many ways there are of choosing $S$, which is $$g + (c+1)d + (b+1)f + (a+1)e + (a+1)(b+1)(c+1),$$ and the proof follows.
\end{proof}

We have the following theorem.
\begin{theorem}\label{thm:m11}
If the initial hypergraph contains at least one hyperedge, then the number of motifs of type 11 in the ILTH model is $\Omega(k^{2t})$.
\end{theorem}
\begin{proof}
A motif of type 11 has cardinality vector $(a,0,c,d,e,0,g)$, where $a+d+g = d+e+g = c+e+g=k$, which yields $a=e$ and $c=d$.
For each motif of type 11 in $H_t$, there will be
\begin{align*}
g + (c+1)d + (a+1)e + (a+1)(c+1) &= g +2a+2c + c^2 + a^2 +ac +1 \\
&= (k-g)^2-ac+2k-g+1,
\end{align*}
motifs of type 11 in $H_{t+1}$, which is maximized when $g=1$ and either $a=1$ or $c=1$ (as $a,c,g>0$) yielding a maximum value $k^2-k+3$.

Let $e_1$ be a hyperedge in $H_0$.
For some $u\in e_1$, there is a hyperedge $e_2 = e_1\cup \{\clone{u}\} \setminus \{u\}$ in $H_1$.
For some $v \in e_1\setminus \{u\}$, there is a hyperedge $e_3$ in $H_{k-2}$ with $e_3\cap e_2 \cap e_1 = \{u\}$ and $e_3\cap e_1 = \{u,v\}$.
These three hyperedges form a motif of type 11 in $H_{k-2}$ with cardinality vector $(1,0,k-2,k-2,1,0,1)$.
As such, by Lemma~\ref{lm:fromCardVect} there are at least $(k^2-k+3)^{t-k+2} = \Omega(k^{2t})$  motifs of type 11 in $H_{t}$ with cardinality vector $(1,0,k-2,k-2,1,0,1)$.
\end{proof}

We can also perform a similar analysis of the other motif types that grow rapidly.
\begin{theorem}\label{mmore}
If the initial hypergraph contains at least one hyperedge, then the number of each motif of types 2, 6, 12, 16, and 26 in the ILTH model is $\Omega(k^{3t})$.
\end{theorem}
\begin{proof}
It is straightforward to verify that $H_{3k}$ contains a motif of type $i$ containing $\alpha_i$ vertices, for $i \in \{2,6,11,12,16\}$. Suppose that the motif in question has cardinality vector $(a,b,c,d,e,f,g)$, and so $\alpha_i = a+b+c+d+e+f+g$.
As $\alpha_i = \Omega(k^3)$ for these values of $i$, by Lemma~\ref{lm:fromCardVect}, there are at least $(\alpha_i)$-times more of this motif type and cardinality vector in each iteration of the ILTH process.
Hence, there are at least $(\alpha_i)^{t-3k} = \Omega(k^{3t})$ of this motif type in $H_t$, and the result follows.
\end{proof}

Our analysis so far does not apply to motif types $13$, $14$, $15$, $24$, and $25,$ as each of these motif types will not be generated in the ILTH process on one hyperedge. However, if one of these motif types occurs within the starting hypergraph, then we will have exponential growth of these, as shown in the following theorem.
\begin{theorem}
If $H_0$ contains a motif of type $i \in \{13,14,15,24,25\}$ that contains $m$ vertices, then motif $i$ occurs at least $(m+1)^t$ times in $H_t$.
\end{theorem}
\begin{proof}
The proof follows by Lemma~\ref{lm:fromCardVect}.
\end{proof}

We contrast the motif counts for ILTH with comparable random hypergraphs. Let $G(n,k,p)$ be the random hypergraph where each possible $k$-set is included as a hyperedge with probability $p$. If we fix two vertices $u$ and $w$, then the expected number of hyperedges $e$ containing both $u$ and $w$ is $\binom{n-2}{k-2}p$.

We consider the $k$-uniform hypergraph with $n = n(t) = 2^t n(0) = \Theta\left(2^t\right)$ vertices and $$p = \frac{e(t)}{\binom{n(t)}{k}}= \Theta\left(2^{\left(\log_2(k+1) - k \right)t}\right).$$
We expect $\Theta(n^{\alpha_i} p^3)$ motifs of type $i$ with $\alpha_i$ vertices. To see this, give each vertex in a motif with $\alpha_i$ vertices a label between $1$ and $\alpha_i$, and define three $k$-sets with these labels $e_1$, $e_2$, and $e_3$ from the three hyperedges in the motif with these labels. We select $\alpha_i$ vertices from the set of $n$ vertices in the $k$-uniform random hypergraph, labeling the $i$th choice by the label $i$. There are $\frac{n!}{(n-\alpha_i)!} \sim n^{\alpha_i}$ possible ways to make these choices.
The sets of vertices $e_1$, $e_2$, and $e_3$ are hyperedges in the $k$-uniform random hypergraph with probability $p^3$. There is systematic double counting of occurrences of the motif but this only changes the expectation by a multiple of some function of $k$, which is a constant. The motifs of type $i$ with fewer than $\alpha_i$ vertices will occur $o(n^{\alpha_i} p^3)$ times, so the total number of motifs of type $i$ that have any number of vertices is  $\Theta(n^{\alpha_i} p^3)$.

Therefore, we expect
\begin{align*}
\Theta(n^{\alpha_i} p^3) = \Theta\left( 2^{\big(\alpha_i - 3\left(k - \log_2(k+1)\right)\big)t} \right)
\end{align*} many occurrences of motif $i$.
If $\alpha_i < 3\left(k - \log_2(k+1)\right),$ then the expected number of motifs of type $i$ will tend to $0$ exponentially fast, and if  $\alpha_i > 3\left(k - \log_2{k+1}\right),$ then the expected number of motifs of type $i$ grows exponentially. In particular, it will be useful to note that  if $\alpha_i \leq 2k-1$ and $k \ge 9$, then the expected number of motifs of type $i$ will tend to $0$ exponentially fast, and if $\alpha_i = 3k-c$ with $c \in \{2, 3,4, 5\}$, then the expected number of motifs of type $i$ grows exponentially fast.

As a consequence, we expect motifs $2$, $6$, $12$, $16$, and $26$ to occur an exponential number of times each in a random hypergraph. We expect that other motifs will rarely occur, with the probability that we see any diminishing when $k\geq 9$ and $t$ increases. As a consequence of Theorems~\ref{thm:m11} and \ref{mmore}, the growth rates of the motif types 2, 6, 11, 12, 16, and 26 is faster in ILTH than in a comparable random $k$-uniform hypergraph.

We finish the section with precise motif counts for ILTH with initial hypergraph a single hyperedge. We ran the ILTH model on a computer, starting with a single hyperedge of cardinality $k$, for $3 \leq k \leq 6$ and $1 \leq t \leq 10-k$.

See Tables~2 to 5 below for the motif counts of these ILTH hypergraphs.

\begin{center}
\begin{table}[ht!]
\begin{tabular}{|l|l|l|l|l|l|l|l|l|l|l|l|l|l|l|}
\hline
t&2&6&11&26\\ \hline
1&&&3&1\\
2&45&126&75&45\\
3&3447&4770&1083&1141\\
4&161451&115146&12675&22365\\
5&5981355&2301930&133563&382981\\
6&195870195&41818266&1326675&6071085\\
7&5993456427&720709290&12718443&91888021\\
\hline
\end{tabular}
\caption{The number of motifs generated by the ILTH model starting with a hyperedge of cardinality $3$.}
\label{tab:ILTHmotifCounts3}
\end{table}

\medskip

\begin{table}[ht!]
\begin{tabular}{|l|l|l|l|l|l|l|l|l|l|l|l|l|l|l|}
\hline
t&2&6&11&12&16&26\\ \hline
1&&&6&&4&\\
2&90&504&474&504&188&276\\
3&16660&75168&14010&42192&5116&34248\\
4&2651330&6088680&305682&1920888&107712&2341332\\
5&305991860&369517680&5764506&67434480&2026684&122766120\\
6&28267339810&19173430584&100158594&2066592024&34911788&1285323380\\
\hline
\end{tabular}
\caption{The number of motifs generated by the ILTH model starting with a hyperedge of cardinality $4$.}
\label{tab:ILTHmotifCounts4}
\end{table}

\begin{table}[ht!]
\begin{tabular}{|l|l|l|l|l|l|l|l|l|l|l|l|l|l|l|}
\hline
t&2&6&11&12&16&26\\ \hline
1&&&10&&10&\\
2&150&1110&1490&2100&1870&420\\
3&40210&356670&82030&540720&189610&234360\\
4&13613610&77687610&3114650&71894820&12725950&50062740\\
5&4067088850&12719703750&97894510&6831291600&680649610&7078307400\\
\hline
\end{tabular}
\caption{The number of motifs generated by the ILTH model starting with a hyperedge of cardinality $5$.}
\label{tab:ILTHmotifCounts5}
\end{table}

\begin{table}[ht!]
\begin{tabular}{|l|l|l|l|l|l|l|l|l|l|l|l|l|l|l|}
\hline
t&2&6&11&12&16&26\\ \hline
0&&&&&&\\
1&&&15&&20&\\
2&229&2070&3285&5040&7680&120\\
3&79096&994680&301515&2610180&1983740&576720\\
4&388621215&409931190&18710325&815537880&346117200&370671840\\
\hline
\end{tabular}
\caption{The number of motifs generated by the ILTH model starting with a hyperedge of cardinality $6$.}
\label{tab:ILTHmotifCounts6}
\end{table}
\end{center}

\newpage

\section{Hypergraph clustering coefficients}\label{secclus}

The small-world property in complex networks demands low average distance and high clustering coefficients, relative to random graphs with the same expected average degree; see \cite{bbook} for a discussion. An analogous definition holds for small-world hypergraphs, comparing their properties to a random hypergraph $G(n,k,p)$ with the same order $n$ and $p$ chosen so that they have the same expected average degree. As we demonstrated in Section~\ref{2sec}, ILTH hypergraphs have constant average distance. Hence, a natural next step in our investigation is to consider clustering coefficients of ILTH hypergraphs.

There are a variety of hypergraph clustering coefficients we may consider; see \cite{GallagherGoldberg2013} for nine distinct coefficients. We focus on a clustering coefficient introduced in \cite{estrada}, along with a new one that is a variant of the one studied in \cite{zn}. We discuss these clustering coefficients by considering graphs. For a graph $G,$ the global clustering coefficient is
$$
C(G) = \frac{6 \times (\text{number of triangles in $G$})}{\text{number of paths of length two in $G$}}.
$$
Note that $C(G)$ is a rational number in the interval $[0,1]$.

There are several different ways to generalize the definition of clustering coefficient to hypergraphs. We discuss three of these in the context of the ILTH model.

We define a \emph{path} of length two in a hypergraph to be a 5-tuple $(u,e_1,v,e_2,w)$ where $u,v,w$ are distinct vertices, $e_1,e_2$ are distinct hyperedges, and $u,v \in e_1$, $v,w \in e_2$.
Similarly, we define a \emph{hypertriangle} to be a 6-tuple $ (u,e_1,v,e_2,w,e_3)$ where $u,v,w$ are distinct vertices, $e_1,e_2,e_3$ are distinct hyperedges, and $u,v \in e_1$, $v,w \in e_2$, $w,u \in e_3$.
We have the following generalization of the clustering coefficient to hypergraphs, appearing first in \cite{estrada}:
$$
\HC_1(H) = \frac{ 6 \times \text{(number of hypertriangles in $H$)}}{\text{number of \hyperpaths  of length two in $H$}}.
$$
Note that $\HC_1(H) = C(H)$ in the case that $H$ is a graph. However, for general hypergraphs $H$, the values of $\HC_1(H)$ need no longer be in the interval $[0,1]$. For example, the complete $k$-uniform hypergraph on $n$ vertices has $\HC_1(GK^{(k)}_n) = \binom{n-2}{k-2}$. The reason for this difference with the graph case is because a given path of length two $(u,e_1,v,e_2,w)$ can be extended to a hypertriangle in many different ways. The hyperedge $e_3$ can be any hyperedge so long as it includes $u$ and $w$. The clustering coefficient $\HC_1$ counts the average number of hypertriangles that are extensions of a path of length two.

We prove the following theorem on $\HC_1$ in Subsection~\ref{sec:HC_1}.
\begin{theorem}\label{thm:HC_1}
For a nonnegative integer $t$, we have that
$$\HC_1(H_t)  = \Theta\left( \left(\frac{(k-1)^3 + 3(k-1)}{k^2 + 1}\right)^t \right).$$
\end{theorem}
We can show that $H_t$ has a higher value of $\HC_1$ than the random $k$-uniform hypergraph with the same number of vertices and the same expected average degree. See the discussion at the end of Subsection~\ref{sec:HC_1}.

There are other ways to express the clustering coefficient on graphs that lead to different generalisations to hypergraphs. One such equivalent definition is that $C$ is the probability that given a path of length two, the end vertices are adjacent:
$$
C(G) = \mathbb{P}\left( uv \text{ is an edge} : (u,e_1,w,e_2,v) \text{ a path of length two}
\right).
$$
\noindent We say two vertices $u,v$ in a hypergraph are \emph{adjacent}, written $u \sim v$, if there is some hyperedge $e$ containing both. There is then a natural way to generalize this definition of $C$ to hypergraphs, which we think we are, surprisingly, the first to propose.
\begin{align*}
\HC_2(H) &= \mathbb{P}\left( u \sim v  : (u,e_1,w,e_2,v) \text{ a \hyperpath of length two in $H$}
\right)\\
&= \frac{\text{number of \hyperpaths  $(u,e_1,w,e_2,v)$, where $u \sim v$} }{\text{number of \hyperpaths  of length two} }.
\end{align*}
Note that since $\HC_2$ is a probability, this clustering coefficient is bounded between $0$ and $1$. Further, $\HC_2$ matches the clustering coefficient $C$ on graphs.

A different generalization of the clustering coefficient to hypergraphs, due to \cite{zn}, also retains the property that the clustering coefficient is between 0 and 1, and is closely related to $\HC_2$. Let $\mathcal{I}$ be the set of pairs of intersecting edges in $H$.
For a $(e,f) \in \mathcal{I}$, define $$A(e,f) = |\{u \in e-f:  \text{ for some } w \in f - e  \text{ with } u \sim w \}|.$$ For $e_1,e_2 \in \mathcal{I}$ define $${\rm EO}(e_1,e_2) = \frac{A(e_1,e_2) + A(e_2,e_1)
}{|e_1 - e_2| + |e_2 - e_1|}.$$
The extra overlap attempts to capture the number of connections between vertices $u \in e_1 - e_2$ and $w \in e_2 -e_1$.
It is evident that $0 \le {\rm EO}(e_1,e_2) \le 1$. The following clustering coefficient from \cite{zn} is the average extra overlap over all intersecting pairs of edges:
\begin{equation*}
\HC_3(H) = \frac{1}{|\mathcal{I}|} \sum_{(e_i,e_j) \in \mathcal{I}} \rm{EO}(e_i,e_j).
\end{equation*}

The goals of the authors in \cite{zn} were to define a clustering coefficient on hypergraphs that i) took values in $[0,1]$, ii) matches the normal clustering coefficient when applied to graphs, and iii) reflects the extent of connectivity among neighbors of $v$ due to hyperedges other than ones connecting $v$ with those neighbors. These three goals are satisfied by $\HC_3$, but they are also all satisfied by $\HC_2$, which we believe to be a more natural definition given that it can be simply expressed as a probability without recourse to the notion of extra overlap. For these reasons, we focus on the new clustering coefficient parameter $\HC_2.$

We prove the following theorem on $\HC_2$ in Subsection~\ref{sechc2}.
\begin{theorem}\label{thm:HC_2}
For a nonnegative integer $t,$ we have that
\begin{equation*}
\HC_2(H_t)  =  \Theta \left( \left( \frac{k^2}{k^2 + 1} \right)^t \right).
\end{equation*}
\end{theorem}

We show that $H_t$ has a lower value of $\HC_2$ than the random $k$-uniform hypergraph with the same number of vertices and the same expected average degree, and so by this measure, it has less clustering. This is in contrast to the clustering coefficient $\HC_1,$ and we include a discussion of this phenomenon at the end of the section. We introduce a modified version of ILTH where
clones and their parents are in certain hyperedges. For the modified ILTH model, $\HC_2$ has higher values than in random hypergraphs.

The following lemma will prove useful in our study of hypergraph clustering coefficients.
\begin{lemma} \label{lem:inclusion}
Suppose that $v \in V(H_{t-1})$ and $e \in E(H_{t-1})$ with $v \not\in e$. Let $v^* \in V(H_t)$ be a descendant of $v$ and $e^* \in E(H_t)$ be a descendant of $e$. We then have that $v^* \not\in e^*$.
\end{lemma}
\begin{proof}
Take some $v \in V(H_{t-1})$ and $e \in E(H_{t-1})$ with $v \not\in e$.  The descendants of $e$ are $e$ and $e - x + x'$ for each $x \in e$. Since $v \not \in e$, it is evident that $v$ and $v'$ are not contained in any of the descendants of $e$.
\end{proof}

Lemma~\ref{lem:inclusion} is more useful for our purpose in its contrapositive form.
\begin{lemma} \label{lem:inclusion1}
Suppose that $v^* \in V(H_t)$ and $e^* \in E(H_t)$ with $v^* \in e^*$. If $v \in V(H_{t-1})$ and $e \in E(H_{t-1})$ are their respective predecessors, then $v \in e$.
\end{lemma}

\subsection{The clustering coefficient $\HC_1$}\label{sec:HC_1}

This subsection is devoted to proving Theorem~\ref{thm:HC_1}. To that end, we prove two combinatorial lemmas finding the asymptotic order of the number of paths of length two and the number of hypertriangles in $H_t$, respectively.

\begin{lemma} \label{lem:paths}
The number of paths of length two in $H_t$ is $\Theta\left( \left( k^2 + 1\right)^t \right)$.
\end{lemma}
\begin{proof}
Let $P'(t) = \{ (e_1,v,e_2) : v \in V(H_t), e_1,e_2 \in E(H_t), v \in e_1 \cap e_2 \}$. Note that, while closely related, this is not the same as the set of paths of length two as we do not include endpoints. We include the degenerate case where $e_1 = e_2$. We find an exact value for $|P'(t)|$ in terms of $t$ and $|P'(0)|$, which will enable us to bound the number of paths of length two.

Fix some $(e_1,v,e_2) \in P'(t-1)$. We wish to count the number of descendants $(e_1^*,v^*,e_2^*)$ this has in $P'(t)$. If $v^* = v',$ then for $v^* \in e_1^* \cap e_2^*$ we must have $e_1^* = e_1 - v + v'$ and $e_2^* = e_2 - v + v'$, so there is one descendant  $(e_1^*,v^*,e_2^*)$ in $P'(t)$, where $v^* = v'$. If $v^* = v$, then for $v \in e_1^* \cap e_2^*$ we cannot have $e_1^* \ne e_1 - v + v'$ and $e_2^* \ne e_2 - v + v'$. All of the $k$ other descendants of $e_1$ and the $k$ other descendants of $e_2$ contain $v$ so there are $k^2$ descendants $(e_1^*,v^*,e_2^*)$ in $P'(t)$, where $v^* = v$.

In total, each $(e_1,v,e_2) \in P'(t-1)$ has $k^2 + 1$ descendants $(e_1^*,v^*,e_2^*)$ in $P'(t)$, giving $$|P'(t)| \ge (k^2+1)|P'(t-1)|.$$

Next, suppose we have some  $(e_1^*,v^*,e_2^*) \in P'(t)$, so in particular, $v^* \in e_1^* \cap e_2^*$. Consider their respective predecessors $e_1,v$, and $e_2$ in $H_{t-1}$. Lemma~\ref{lem:inclusion1} provides that $v \in e_1$ and $v \in e_2$, so $(e_1,v,e_2) \in P'(t-1)$. Hence, every triple in $P'(t)$ is a descendant of a triple in $P'(t-1)$, and in particular, $|P'(t)| = (k^2+1)|P'(t-1)|$. Iterating this process, we derive that $|P'(t)| = \left(k^2 + 1\right)^t |P'(0)|$.

Now, let $P(t)$ be the set of paths of length two in $H_t$. Recall that a path of length two is $(u,e_1,v,e_2,w)$ where $u,v,w \in V(H_t)$ are distinct, $e_1,e_2 \in E(H_t)$ are distinct, and $u, v \in e_1$, $v,w \in e_2$.

For $0 \le i \le k$, let $P_i(t)$ be the set of ordered pairs $(e_1,e_2)$ of hyperedges with $|e_1\cap e_2| = i$. Note that $|P_k(t)| = e(t)$.
We then have that $|P'(t)| = \sum_{i=1}^k i|P_i(t)|$ and $P(t) = \sum_{i=1}^{k-1} i(k-i)^2 |P_i(t)|$.
This gives that
\begin{align*}
|P'(t)| - k e(t) = \sum_{i=1}^{k-1} i|P_i(t)| \le |P(t)|  &\le (k-1)^2 |P'(t)| \\
\left(k^2+1\right)^t|P'(0)| - k (k+1)^te(0) \le|P(t)| &\le (k-1)^2\left(k^2+1\right)^t|P'(0)|,
\end{align*}
which completes the proof.
\end{proof}

We next have the following lemma.
\begin{lemma} \label{lem:triangles}
The number of hypertriangles in $H_t$ is $\Theta\left( \left((k-1)^3 + 3(k-1)\right)^t\right)$.
\end{lemma}
\begin{proof}
Let
\begin{align*}
T'(t) = \left\{
\begin{array}{ll}
(u,e_1,v,e_2,w,e_3) :& u,v,w \in V(H_t) \text{ distinct}, e_1,e_2, e_3 \in E(H_t) \\
& u \in e_1\cap e_3, v \in e_1 \cap e_2, w \in e_2 \cap e_3
\end{array}
\right\} .
\end{align*}
Note that, while closely related, this is not the same as the set of hypertriangles as we do not insist that the edges $e_1,e_2$ and $e_3$ are distinct. We find an exact value for $|T'(t)|$ in terms of $t$ and $|T'(0)|$, which will enable us to bound the number of hypertriangles.

Fix some $(u,e_1,v,e_2,w,e_3)  \in T'(t-1)$. We wish to count the number of descendants $$(u^*,e_1^*,v^*,e_2^*,w^*,e_3^*) $$ has in $T'(t)$. If $v^* = v',$ then for $v^* \in e_1^* \cap e_2^*$ we must have $e_1^* = e_1 - v + v'$ and $e_2^* = e_2 - v + v'$. Since $u' \not\in e_1 - v + v'$ and $w' \in e_2 - v + v'$ this means that $u^* = u$ and $w^* = w$. Since $u^*$ and $w^*$ are in $e_3^*$, $e_3^*$ must be $e_3$ or $e_3 - x + x'$ for some $x \in e_3$ not equal to $u$ or $w$, and indeed each of these $k-1$ choices for $e_3^*$ gives a $(u^*,e_1^*,v^*,e_2^*,w^*,e_3^*) $ in $T'(t)$.

An analogous argument in the cases $u^* = u'$ and $w^* = w'$ show that if one of $u^*,v^*,w^*$ is a clone then the other two are not, and there are $3(k-1)$ descendants $(u^*,e_1^*,v^*,e_2^*,w^*,e_3^*) $ in $T'(t)$ of this form.

Otherwise, none of $u^*,v^*,w^*$ is a clone. We then have that $e_1^*$ must be $e_1$ or $e_1 - x + x'$ for some $x \in e_1 - u - v$, $e_2^*$ must be $e_2$ or $e_2 - y + y'$ for some $y \in e_2 - v - w$, and $e_3^*$ must be $e_3$ or $e_3 - z + z'$ for some $z \in e_3 - u - w$. Any combination of these gives a $(u^*,e_1^*,v^*,e_2^*,w^*,e_3^*) $ in $T'(t)$, and so there are $(k-1)^3$ contributing to the count. In total, each $(u,e_1,v,e_2,w,e_3)  \in T'(t-1)$ has $(k-1)^3 + 3(k-1)$ descendants $(u^*,e_1^*,v^*,e_2^*,w^*,e_3^*) $ in $T'(t)$ giving $|T'(t)| \ge \left((k-1)^3 + 3(k-1)\right)|T'(t-1)|$.

In the other direction, suppose we have some  $(u^*,e_1^*,v^*,e_2^*,w^*,e_3^*) $ in $T'(t)$. Consider their respective predecessors $u,e_1,v,e_2,w$ and $e_3$ in $H_{t-1}$. We know that $u,v,w$ must be distinct: if say $u = v$ then either $u^* = v^*$, contradicting that $(u^*,e_1^*,v^*,e_2^*,w^*) \in T'(t)$, or $\{u^*,v^*\} = \{v,v'\}$. This in turn contradicts that there is a hyperedge $e_1^*$ containing both. An analogous argument shows that $v \ne w$ and $w\ne u$. Lemma~\ref{lem:inclusion1} provides that $u \in e_1 \cap e_3$, $v \in e_1 \cap e_2$ and $w \in e_2 \cap e_3$, so $(u,e_1,v,e_2,w,e_3)  \in T'(t-1)$. Hence, every 6-tuple in $T'(t)$ is a descendant of a 6-tuple in $T'(t-1)$, and in particular, $|T'(t)| = \left((k-1)^3 + 3(k-1)\right)|T'(t-1)|$.
Iterating this, we obtain that $|T'(t)| = \left((k-1)^3 + 3(k-1)\right)^t|T'(0)|.$

Now, let $T(t)$ be the set of hypertriangles in $H_t$. Note that $|T(t)|$ is the number of 6-tuples $(u,e_1,v,e_2,w,e_3)$ in $T'(t),$ where $e_1,e_2$ and $e_3$ are all distinct. Hence, we have that $$|T(t)| \le |T'(t)| = \left((k-1)^3 + 3(k-1)\right)^t|T'(0)|.$$

For a lower bound, we count the number of 6-tuples where $e_1, e_2$ and $e_3$ are not distinct.
If $e_1 = e_2 \ne e_3$, then $u$ and $w$ are distinct elements in $e_1 \cap e_3 = e_2 \cap e_3$ and $v \in e_1 - u - w$. Recalling that $|P_i(t)|$ is the number of pairs of edges intersecting in $i$ vertices, we find that there are $\sum_{i=2}^{k-1} i(i-1)(k-2)|P_i(t)|$ such 6-tuples. Similarly, there are $\sum_{i=2}^{k-1} i(i-1)(k-2)|P_i(t)|$ with $e_2 = e_3 \ne e_1$ and with $e_3 = e_1 \ne e_2$.

Finally, note that when $e_1 = e_2 = e_3$ then we just have $u,v,w$ distinct vertices in $e_1$ and so there are $k(k-1)(k-2)e(t)$ 6-tuples in $T'(t)$ with $e_1 = e_2 = e_3$. Putting these together gives $|T'(t)| = |T(t)| + 3\sum_{i=2}^{k-1} i(i-1)(k-2)P_i(t) + k(k-1)(k-2)e(t)$.

To bound $\sum_{i=2}^{k-1} i(i-1)|P_i(t)|$, we use that $ \sum_{i=1}^k i|P_i(t)| = |P'(t)| = (k^2+1)^t|P'(0)|$ as calculated in the proof of Lemma~\ref{lem:paths}.
In particular, we have that
\begin{equation*}
  \sum_{i=2}^{k-1} i(i-1)|P_i(t)| \le  (k-2)\left(\sum_{i=1}^k i|P_i(t)| \right)
   \le (k-2) (k^2+1)^t|P'(0)|.
\end{equation*}

We next have that
\begin{align*}
|T(t)| &= |T'(t)| - 3(k-2)\sum_{i=2}^{k-1} i(i-1)|P_i(t)| - k(k-1)(k-2)e(t) \\
&\ge \left((k-1)^3 + 3(k-1)\right)^t|T'(0)| - 3(k-2)^2(k^2+1)^t|P'(0)|   - k(k-1)(k-2)(k+1)^te(0),
\end{align*}
which completes the proof.
\end{proof}
As an immediate consequence of Lemmas~\ref{lem:paths} and \ref{lem:triangles}, we obtain Theorem~\ref{thm:HC_1} on the value of the $\HC_1$ clustering coefficient on ILTH hypergraphs. To contextualize the result of Theorem~\ref{thm:HC_1}, we compare $\HC_1(H_t)$ to $\HC_1$ for other $k$-uniform hypergraphs. For the complete $k$-uniform hypergraph $K^{(k)}_n$ it is straightforward to derive by counting choices of $u,v,w$ and the edges containing them that
$$\HC_1\left(K^{(k)}_n\right) = \frac{ \binom{n}{3}\left(\binom{n-2}{k-2}\right)^3}{ \binom{n}{3}\left(\binom{n-2}{k-2}\right)^2} = \binom{n-2}{k-2}.$$
When $n = n(t) = 2^tn(0)$, this gives $\HC_1(K^{(k)}_n) = \Theta \left( 2^{(k-2)t} \right)$, which is larger than $\HC_1(H_t) ,$ as expected.

We consider the expected value of $\HC_1$ in the random hypergraph $G(n,k,p)$. Here, given a path $(u,e_1,v,e_2,w)$  of length two, the expected number of hypertriangles of the form $(u,e_1,v,e_2,w,e)$ is $\binom{n-2}{k-2}p$ . This gives
$$\mathbb{E}\left( \HC_1(G(n,k,p)) \right)= \binom{n-2}{k-2}p.$$

Let $n = n(t)= 2^tn(0)$ and $p = \frac{ (k+1)^te(0)}{\binom{n}{k}}$. We then have that $$\mathbb{E}\left( \HC_1(G(n,k,p)) \right) =  \frac{ \binom{n-2}{k-2} (k+1)^te(0)}{ \binom{n}{k}} = \frac{k(k-1) (k+1)^te(0) }{2^tn(0)(2^tn(0)-1)}= \Theta\left( \left( \frac{k+1}{4}\right)^t \right).$$
As $\frac{k+1}{4} < \frac{(k-1)^3 + 3(k-1)}{k^2 + 1}$, the clustering coefficient $\HC_1$ for $H_t$ grows faster than that for the random hypergraph of the same expected average degree.

\subsection{The clustering coefficient $\HC_2$}\label{sechc2}
In this subsection, we prove Theorem~\ref{thm:HC_2}. We first introduce a useful set of 5-tuples:
\begin{align*}
A(t) = \left\{
\begin{array}{ll}
(u,e_1,v,e_2,w) :& u,v,w \in V(H_t) \text{ distinct}, e_1,e_2 \in E(H_t), \\
& \text{ for some } e_3  \in E(H_t) \text{ such that } u \in e_1\cap e_3, v \in e_1 \cap e_2, w \in e_2 \cap e_3
\end{array}
\right\}.
\end{align*}
One view of a 5-tuple in $A(t)$ is as a (possibly degenerate) path of length 2 that can be completed to a (possibly degenerate) hypertriangle. We have the following lemma counting
the elements of $A(t)$, which will greatly assist in estimating $\HC_2$ in ILTH hypergraphs.

\begin{lemma} \label{lem:A}
For all nonnegative integers $t,$ $|A(t)| =  \left(k^2\right)^t |A(0)|.$
\end{lemma}
\begin{proof}
For a fixed 5-tuple $(u,e_1,v,e_2,w)  \in A(t-1)$, we count the number of descendants $(u^*,e_1^*,v^*,e_2^*,w^*) $ this has in $A(t)$. If $v^* = v'$, then for $v^* \in e_1^* \cap e_2^*$ we must have $e_1^* = e_1 - v + v'$ and $e_2^* = e_2 - v + v'$. Since $u' \not\in e_1 - v + v'$ and $w' \in e_2 - v + v'$ this means that $u^* = u$ and $w^* = w$, and we know there is a hyperedge $e_3$ containing both. Thus, there is one descendant  $(u^*,e_1^*,v^*,e_2^*,w^*) $ in $A(t)$ with $v^* = v'$.

Otherwise, suppose $v^*=v$. We cannot have both $u^* = u'$ and $w^* = w'$ as there does not exist any hyperedge in $E(H_t)$ containing both $u'$ and $w'$.
We can have $u^* = u'$ and $w^* = w$, as the hyperedge $e_3 - u + u' \in e(H_t)$ contains both. In this case, $e_1^*$ must be $e_1 - u + u'$  and $e_2^*$ must be $e_2$ or $e_2 - y + y'$ for some $y \in e_2 - v - w$, giving $k-1$ descendants in $A(t)$.
Similarly, we can have $u^* = u$ and $w^* = w'$, and there are a further $k-1$ descendants in $A(t)$ of this form.

Finally, we can have $u^* = u$ and $w^*=w$ as we know the hyperedge $e_3$ contains both. In this case $e_1^*$ must be $e_1$ or $e_1 - x + x'$ for some $x \in e_1 - v - u$  and $e_2^*$ must be $e_2$ or $e_2 - y + y'$ for some $y \in e_2 - v - w$, giving $(k-1)^2$ descendants in $A(t)$ of this form.
In total, each $(u,e_1,v,e_2,w)  \in A(t-1)$ has $k^2$ descendants $(u^*,e_1^*,v^*,e_2^*,w^*) $ in $A(t)$ giving $|A(t)| \ge k^2|A(t-1)|$.

In the other direction, suppose we have some  $(u^*,e_1^*,v^*,e_2^*,w^*) $ in $A(t)$. Let $e_3^*$ be a hyperedge in $H_t$ containing both $u^*$ and $w^*$. Consider their respective predecessors $u,e_1,v,e_2,w$ and $e_3$ in $H_{t-1}$. We know that $u,v,w$ must be distinct: if say $u = v$ then either $u^* = v^*$, contradicting   $(u^*,e_1^*,v^*,e_2^*,w^*) \in A(t)$, or $\{u^*,v^*\} = \{v,v'\}$, contradicting that there is a hyperedge $e_1^*$ containing both. An analogous argument shows that $v \ne w$ and $w\ne u$.  Applying Lemma~\ref{lem:inclusion} shows that $u \in e_1 \cap e_3$, $v \in e_1 \cap e_2$ and $w \in e_2 \cap e_3$, so $(u,e_1,v,e_2,w)  \in A(t-1)$.
Hence, every 5-tuple in $A(t)$ is a descendant of a 5-tuple in $A(t-1)$, and in particular, $|A(t)| = k^2|A(t-1)|$. Iterating this, we have that $|A(t)| = \left(k^2\right)^t|A(0)|.$
\end{proof}

We can now use Lemma~\ref{lem:A} to prove Theorem~\ref{thm:HC_2}.

\begin{proof}[Proof of Theorem~\ref{thm:HC_2}]
Recall that
\begin{align*}
\HC_2(H) = \frac{\text{number of \hyperpaths  $(u,e_1,w,e_2,v)$, where $u$ and $v$ are in a hyperedge} }{\text{number of \hyperpaths  of length two} }.
\end{align*}
Let $\Lambda(t)$ be the number of \hyperpaths  $(u,e_1,v,e_2,w),$ where $u \sim w$. We then have that $\Lambda(t) \subseteq A(t)$. Also, a 5-tuple $(u,e_1,v,e_2,w)$ is in $A(t)$ but not $\Lambda (t)$ if and only if $e_1 = e_2,$ and there are $k(k-1)(k-2)e(t)$ such 5-tuples.
Thus, we have that $$|\Lambda(t)| = |A(t)| - k(k-1)(k-2)e(t) =  k^{2t}|A_0| - k(k-1)(k-2)(k+1)^te(0) = \Theta \left( k^{2t} \right).$$

Combining this with Lemma~\ref{lem:paths}, we derive that
$
\HC_2 = \Theta  \left(  \left( \frac{k^2}{k^2 + 1} \right)^t \right),
$
as required.
\end{proof}

We contextualize these results by comparing them to the random $k$-uniform hypergraph $G(n,k,p)$ with the same expected average degree. We derive a lemma computing the expected value of
$\HC_2$ on random hypergraphs.

\begin{lemma}\label{lem:HC_2 random}
For a given $k$ and $p$, we have that
$$\mathbb{E}\left(\HC_2(G(n,k,p))\right) = 1 - (1-p)^{\binom{n-2}{k-2}}.$$
\end{lemma}
\begin{proof}
Suppose that we are given a \hyperpath $(u,e_1,v,e_2,w)$ and we wish to know the probability that the two vertices $u,w$ lie in some hyperedge. There are $\binom{n-2}{k-2}$ $k$-sets containing both $u$ and $w$ and the probability that none of them is a hyperedge  of $G(n,k,p)$ is $(1-p)^{\binom{n-2}{k-2}}$. Thus, the probability that $u \sim w$ is $1 - (1-p)^{\binom{n-2}{k-2}}$.
\end{proof}

We  compare $H_t$ to a random hypergraph with the same number of vertices and the same expected average degree. Set $n = 2^tn(0)$ and choose $p$ such that $\binom{n}{k}p = (k+1)^te(0)$.
We then have that
\begin{align*}
\mathbb{E}\left(\HC_2(G(n,k,p))\right) \ge  1 - (1-p)^{\binom{n-2}{k-2}} &\ge 1 - \exp{\left(-p\binom{n-2}{k-2}\right)} \\
 &\ge 1 - \exp{\left(-c\left(\frac{k+1}{4}\right)^t\right)},\end{align*}
where $c$ depends only on $k,$ $n(0),$ and $e(0)$. Hence, we conclude that $\mathbb{E}\left(\HC_2(G(n,k,p))\right)$ is at least $1 - \exp{\left(-c\left(\frac{k+1}{4}\right)^t\right)}$.
For $k\ge 4$, this quantity tends to 1 as $t$ tends to infinity, and it does so doubly exponentially fast. On the other hand, we have that $\HC_2(H_t) =O\left( \left(\frac{k^2}{k^2+1} \right)^t \right)$ which tends to $0$ exponentially fast as $t$ tends to infinity. Thus, we find that by this measure the clustering for $H_t$ is extremely low compared to the random hypergraph with the same expected average degree. If $k=3$, then $\mathbb{E}\left(\HC_2(G(n,k,p))\right)$ is at least the constant $1 - e^{-c}$, which is larger than $\HC_2(H_t)=O \left( \left(\frac{9}{10} \right)^t \right)$.

Measured by the clustering coefficient $\HC_1$, the hypergraph $H_t$ has higher clustering than in comparable hypergraphs, but this fails for $\HC_2$. The reason for the discrepancy is that the two clustering coefficients are counting different structures. Given a pair of intersecting edges $e_1,e_2$, the value of $\HC_2$ counts how many pairs of vertices $u \in e_1 - e_2, w \in e_2-e_1$ there are that are contained in some hyperedge $e_3$. As this is low for $H_t$ compared to random hypergraphs, fewer of those pairs are contained in any hyperedge than we might expect. The value of $\HC_1$ roughly counts how many edges $e_3$ intersect both $e_1$ and $e_2$ to make a hypertriangle. As this is large for $H_t$ when compared to random hypergraphs, there are more of these edges than we might expect. Hence, relative to the random hypergraph, fewer pairs of vertices  $u \in e_1 - e_2, w \in e_2-e_1$ are contained in a hyperedge, but those that are contained in an hyperedge must be contained in many hyperedges.

\subsection{A variant of ILTH with large $\HC_2$ values}
To remedy the situation with ILTH having lower $\HC_2$ values than random hypergraphs, we consider a variant of the model where clones and their parents are in certain hyperedges. Such a variant is a natural one, as we may expect newly formed hyperedges to include both parent and child vertices.

Let $H^{(2)}_0$ be a fixed $k$-uniform hypergraph and we iteratively construct $H^{(2)}_t,$ where $t\ge 1$ as follows. Suppose that we have $H^{(2)}_t$. For each $v \in V(H^{(2)}_t)$, add $k$ vertices $v$ and $v^1,v^2,\dots, v^{k-1}$ to $H^{(2)}_{t+1}$. We call these $v^i$ the \emph{clones} of $v$. For each $e \in E(H^{(2)}_t)$, add to $H^{(2)}_{t+1}$ the hyperedge $e$ and each of the edges $e - x + x^i,$ where $x$ is a vertex in $e$ and $1 \le i \le k-1$. In addition, for each $v \in V(H^{(2)}_t)$ add to $H^{(2)}_{t+1}$ the hyperedge $\{v,v^1,v^2,\dots,v^{k-1}\}$ to $H^{(2)}_{t+1}$. We refer to the model as \emph{ILTH}$_2$, and hypergraphs generated by the model are \emph{ILTH$_2$ hypergraphs}. See Figure~\ref{h2}.
\begin{figure}[h]
\begin{center}
\epsfig{figure=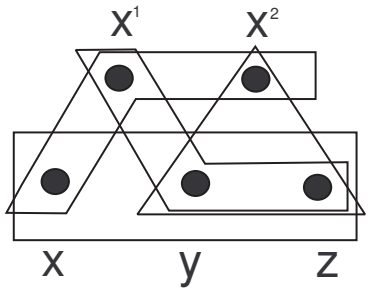, scale=1.5}
\caption{The ILTH$_2$ model applied to cloning $x$ in the hyperedge $xyz$.}\label{h2}
\end{center}
\end{figure}
The ILTH$_2$ model is motivated by the desire to have clones and parent adjacent, as in the original ILT model. While the models are distinct,  ILTH$_2$ hypergraphs share properties with the ILTH hypergraphs such as densification and low distances. One key difference between ILTH and ILTH$_2$ is the clustering coefficient $\HC_2.$ We have the following theorem, whose proof is analogous to the one of Theorem~\ref{thm:HC_2} and so is omitted.

\begin{theorem}\label{thm:HC_2 second model}
For nonnegative integers $t,$ we have that
$$
\HC_2(H^{(2)}_t) = \Theta\left( \left(  1 -  \frac{(k-1)^2}{(k^2 - 2k+1)^2 + k-1}  \right)^t \right).
$$
\end{theorem}

We compare $H^{(2)}_t$ to the random hypergraph with the same number of vertices and the same expected averaged degree. Set $n = k^tn(0)$ and choose $p$ such that the expected number of edges $\binom{n}{k}p$ is $e(t)$. In particular, we have that $$p = \Theta\left( \left(\frac{k^2 - k + 1}{k^k}\right)^t \right).$$

Applying Lemma~\ref{lem:HC_2 random}, we have that $$\mathbb{E}\left(\HC_2(G(n,k,p))\right)= \Theta\left( \left(\frac{k^2 - k + 1}{k^2}\right)^t \right).$$ For all $k \ge 2$, we find that
$$ \frac{k^2 - k + 1}{k^2} =  1 - \frac{k-1}{k^2} < 1 - \frac{(k-1)^2}{(k^2 - 2k + 2)^2 + k-1},$$
so the clustering coefficient $\HC_2$ is larger for $H^{(2)}_t$ than in random hypergraphs.

\section{Further directions}

We introduced the new ILTH model for complex hypergraphs. We found that ILTH hypergraphs densify over time and have low average distances. We considered motifs and found that for those occurring in the ILTH model, their counts grow faster than in random hypergraphs with the same expected average degree. The 2-sections of ILTH hypergraphs were shown to contain isomorphic copies of all graphs admitting a homomorphism to the 2-section of $H_0$ in Theorem~\ref{mainin}. We finished with an analysis of clustering coefficients, and it was shown that $\HC_1$ was larger in ILTH hypergraphs than in random hypergraphs. A similar result was proven for $\HC_2$ applied to a variant of ILTH, where parents are adjacent to their clones.

Several questions remain surrounding ILTH hypergraphs. We may consider variants of the model, and study properties of hypergraphs generated by the model. For example, we may allow hyperedges that are non-uniform orders, or randomize the model by adding random hyperedges to sets of clones. An open problem is to determine the age of ILTH hypergraphs; that is, what are the induced subhypergraphs of ILTH hypergraphs?

Another direction is to consider other notions of clustering in ILTH hypergraphs. Several hypergraph clustering coefficients were investigated in [17], for example, and it would be interesting to consider their values in the ILTH model.


\begin{thebibliography}{99}

\bibitem{aksoy} S.G.\ Aksoy, C.\ Joslyn, C.O.\ Marrero, B.\ Praggastis, E.\ Purvine, Hypernetwork science via high-order hypergraph walks, \emph{EPJ Data Science} \textbf{9} 16, 2020.

\bibitem{benson} A.R.\ Benson, R.\ Abebe, M.T.\ Schaub, A.\ Jadbabaie, J.\ Kleinberg, Simplicial closure and higher-order
link prediction, \emph{Proceedings of the National Academy of Sciences} \textbf{115} (2018) E11221--E11230.

\bibitem{news} A.R.\ Benson, D.F.\ Gleich, D.J.\ Higham, Higher-order network analysis takes off, fueled by old ideas and new data, \emph{SIAM News}, \url{https://cutt.ly/gkwhM9w}, last accessed January 29, 2021.

\bibitem{benson1} A.R.\ Benson, D.F.\ Gleich, J.\ Leskovec, Higher-order organization of complex networks, \emph{Science} \textbf{353} (2016) 163--166.

\bibitem{berge} C.\ Berge, \emph{Graphs and Hypergraphs}, Elsevier, New York, 1973.

\bibitem{berge1} C.\ Berge, \emph{Hypergraphs: The Theory of Finite Sets}, North-Holland, Amsterdam, 1989.

\bibitem{bbook} A.\ Bonato, \emph{A Course on the Web Graph}, American Mathematical Society, Providence, Rhode Island, 2008.

\bibitem{ilm}  A.\ Bonato, H.\ Chuangpishit, S.\ English, B.\ Kay, E.\ Meger, The iterated local model for social networks, \emph{Discrete Applied Mathematics} \textbf{284} (2020) 555--571.

\bibitem{ildt} A.\ Bonato, D.W.\ Cranston, M.A.\ Huggan, T.\ Marbach, R.\ Mutharasan, The Iterated Local Directed Transitivity model for social networks, In: \emph{Proceedings of WAW'20}, 2020.

\bibitem{mgeop} A.\ Bonato, D.F.\ Gleich, M.\ Kim, D.\ Mitsche, P.\ Pra\l{}at, A.\ Tian, S.J.\ Young, Dimensionality matching of social networks using motifs and eigenvalues, \emph{PLOS ONE} \textbf{9}(9):e106052, 2014.

\bibitem{ilt} A.\ Bonato, N.\ Hadi, P.\ Horn, P.\ Pra\l{}at, C.\ Wang, Models of on-line social networks, \emph{Internet Mathematics} \textbf{6} (2011) 285--313.

\bibitem{ilt1} A.\ Bonato, N.\ Hadi, P.\ Pra\l{}at, C.\ Wang, Dynamic models of on-line social networks, In: \emph{Proceedings of WAW'09}, 2009.

\bibitem {bt} A.\ Bonato, A.\ Tian, Complex networks and social networks, invited book chapter in: \emph{Social Networks}, editor E. Kranakis, Springer, Mathematics in Industry series, 2011.

\bibitem{chung1} F.R.K.\ Chung, L.\ Lu, \emph{Complex Graphs and Networks}, American Mathematical Society, Providence, Rhode Island, 2006.

\bibitem{do} M.T.\ Do, S.\ Yoon, B.\ Hooi, K.\ Shin, Structural patterns and generative models of
real-world hypergraphs. In: \emph{Proceedings of  Knowledge Discovery in Databases (KDD)}, 2020.

\bibitem{ek} D.\ Easley, J.\ Kleinberg, \emph{Networks, Crowds, and Markets Reasoning about a Highly Connected World}, Cambridge University Press, 2010.

\bibitem{estrada}  E.\ Estrada, J.A.\ Rodr\'iguez-Vel\'azquez,  Subgraph centrality and clustering in complex hyper-networks, \emph{Physica A: Statistical Mechanics and its Applications} \textbf{364} (2006) 581--594.

\bibitem{GallagherGoldberg2013} S.R.\ Gallaher, D.S.\ Goldberg, \emph{Clustering Coefficients in Protein Interaction
Hypernetworks}, BCB'13: Proceedings of the International Conference on Bioinformatics, Computational Biology and Biomedical Informatics (2013).

\bibitem{he} F.\ Heider, \emph{The Psychology of Interpersonal Relations}, John Wiley \& Sons, 1958.

\bibitem{hn} P.\ Hell, J.\ Ne\v{s}et\v{r}il, \emph{Graphs and Homomorphisms}, Oxford University Press, New York, 2004

\bibitem{lee} G.\ Lee, J.\ Ko, K.\ Shin, Hypergraph motifs: concepts, algorithms, and discoveries, In: \emph{Proceedings of the VLDB Endowment}, 2020.

\bibitem{les1} J.\ Leskovec, J.\ Kleinberg, C.\ Faloutsos, Graphs over time: densification laws, shrinking diameters and possible explanations,
In: \emph{Proceedings of the 13th ACM SIGKDD International Conference on Knowledge Discovery and Data Mining}, 2005.

\bibitem{pr} V. Memi\v{s}evi\'c, T. Milenkovi\'c, N. Pr\v{z}ulj, An integrative approach to modeling biological networks, \emph{Journal
of Integrative Bioinformatics} \textbf{7}:120, 2010.

\bibitem{scott} J.P.\ Scott, \emph{Social Network Analysis: A Handbook}, Sage Publications Ltd, London, 2000.

\bibitem{mason} L.\ Small, O.\ Mason, Information diffusion on the iterated local transitivity model of online social networks, \emph{Discrete Applied Mathematics} \textbf{161} (2013) 1338–1344.

\bibitem{west} D.B.\ West, \emph{Introduction to Graph Theory, 2nd edition}, Prentice Hall, 2001.

\bibitem{zn} W.\ Zhou, L.\ Nakhleh, Properties of metabolic graphs: biological organization in representation artifacts, \emph{BMC Bioinformatics} \textbf{12.132} (2011).

\end{thebibliography}
\end{document}